\newif\ifshortver
\shortverfalse

\newif\iflongver
\ifshortver
\longverfalse
\else
\longvertrue
\fi

\newif\iftodiscard
\todiscardtrue

\documentclass{article}
\usepackage{spconf}

\usepackage[T1]{fontenc}
\usepackage[latin9]{inputenc}
\usepackage{float}
\usepackage{amsmath}
\usepackage{amsthm}
\usepackage{amssymb}
\usepackage{graphicx}
\usepackage{esint}

\makeatletter

\floatstyle{ruled}
\newfloat{algorithm}{tbp}{loa}
\providecommand{\algorithmname}{Algorithm}
\floatname{algorithm}{\protect\algorithmname}

\theoremstyle{plain}
\newtheorem{thm}{\protect\theoremname}
\theoremstyle{definition}
\newtheorem{defn}[thm]{\protect\definitionname}
\theoremstyle{plain}
\newtheorem{prop}[thm]{\protect\propositionname}
\theoremstyle{plain}
\newtheorem{lem}[thm]{\protect\lemmaname}
\theoremstyle{definition}
\newtheorem{rem}[thm]{\protect\remarkname}

\newcommand{\qs}{\delta}

\usepackage{mathrsfs}
\usepackage{algorithm,algpseudocode}

\usepackage{colortbl}
\definecolor{lightgray}{rgb}{0.9,0.9,0.9}
\definecolor{lightred}{rgb}{1,0.8,0.8}
\definecolor{lightgreen}{rgb}{0.6,1,0.6}
\definecolor{lightyellow}{rgb}{1,1,0.5}
\definecolor{lightgrey}{rgb}{0.8,0.8,0.8}

\allowdisplaybreaks[1]

\makeatother

\usepackage[british]{babel}
\providecommand{\definitionname}{Definition}
\providecommand{\propositionname}{Proposition}
\providecommand{\lemmaname}{Lemma}
\providecommand{\theoremname}{Theorem}
\providecommand{\remarkname}{Remark}

\title{Communication-Efficient Laplace Mechanism for Differential Privacy via Random Quantization}

\twoauthors
 {Ali Moradi Shahmiri}
	{Sharif University of Technology\\
	Tehran, Iran}
 {Chih Wei Ling, Cheuk Ting Li\sthanks{The work described in this paper was partially supported by an ECS grant from the Research Grants Council of the Hong Kong Special Administrative Region, China {[}Project No.: CUHK 24205621{]}.}}
	{The Chinese University of Hong Kong\\
 Hong Kong, China}

\begin{document}
\maketitle

\begin{abstract}
We propose the first method that realizes the Laplace mechanism exactly (i.e., a Laplace noise is added to the data) that requires only a finite amount of communication (whereas the original Laplace mechanism requires the transmission of a real number) while guaranteeing privacy against the server and database. Our mechanism can serve as a drop-in replacement for local or centralized differential privacy applications where the Laplace mechanism is used. Our mechanism is constructed using a random quantization technique. Unlike the simple and prevalent Laplace-mechanism-then-quantize approach, the quantization in our mechanism does not result in any distortion or degradation of utility. Unlike existing dithered quantization and channel simulation schemes for simulating additive Laplacian noise, our mechanism guarantees privacy not only against the database and downstream, but also against the honest but curious server which attempts to decode the data using the dither signals.
\end{abstract}
\section{Introduction}

In differential privacy, the Laplace mechanism~\cite{Dwork06DP,Dwork14Book} ensures the privacy of the data by adding a Laplace-distributed noise.
It is widely used in federated learning
\ifshortver
\else
\cite{el2022differential,li2020federated},
\fi
and has been deployed to provide local differential privacy for commercial database management systems (e.g. \cite{kessler2019sap}).

The Laplace mechanism satisfies $d$-privacy or metric privacy \cite{Chatzikokolakis13Dprivacy}, an extension of $\epsilon$-differential privacy~\cite{Dwork06DP,Dwork14Book} to metric spaces. While $d$-privacy was originally studied in a centralized differential privacy setting, it can also be applied to local differential privacy~\cite{andres2013geo,alvim2018local}. 
For instance, in geo-indistinguishability \cite{andres2013geo}, the Laplace mechanism is used to allow users to send their approximate locations to the server, while masking their exact locations.
Also see~\cite{dewri2012local} where the Laplace mechanism is used for location privacy.
It was shown in~\cite{koufogiannis2015optimality} that among all additive noise mechanisms under a given $d$-privacy (or Lipschitz privacy) constraint, the Laplace mechanism attains the optimal utility measured by mean-squared error. 
\ifshortver
\else
Also see~\cite{wang2014entropy,fernandes2021laplace} for other results on the optimality of the Laplace mechanism. This motivates us to focus on schemes that realize the Laplace mechanism.\footnote{It was pointed out in \cite{geng2014optimal} that the Laplace mechanism is not optimal for $\epsilon$-differential privacy for a fixed sensitivity, though the $d$-privacy / Lipschitz privacy constraint~\cite{koufogiannis2015optimality} has an advantage that it can give a differential privacy guarantee that adapts to the sensitivity.}
\fi

The output of the Laplace mechanism is a real number. To transmit this number, a quantization is usually performed, which introduces a distortion and degrades the utility. The overall (Laplace+quantization) noise will not be Laplace.


This paper proposes a quantization method 
that realizes the Laplace mechanism
without any distortion.
To motivate our result, we discuss a simplified local differential privacy scenario (though our method applies to any local/centralized differential privacy scenario where the Laplace mechanism is used).
The user has the data $X$ (e.g. the user's location as in geo-indistinguishability \cite{andres2013geo}),
and wants to convey a noisy version of $X$ to the decoder (or server). The user
and the decoder are allowed to have shared randomness $U$ (e.g. by
sharing a random seed ahead of time). The user sends a compressed description
$M$ with a finite number of bits (using $X$, $U$ and the user's local randomness) to the decoder.
The decoder then recovers $\hat{X}$ (using $U,M$). The goal is to
attain:
\vspace{-5pt}
\begin{itemize}
\item \textbf{Communication cost:} The size of $M$ should be small.
\item \textbf{Decoder privacy:} The data $X$ is kept private from the decoder/server, i.e.,
the leakage of the information in $X$ by $(U,M)$ is small.
\item \textbf{Database privacy:} The data $X$ is kept private from the database (and downstream parties) which only stores $\hat{X}$. The leakage of the
information in $X$ by $\hat{X}$ is small.
\item \textbf{Database utility:} $\hat{X}$ should be close to $X$,
measured by the mean-squared error.
\end{itemize}

\iflongver
\medskip{}
\fi

\ifshortver
\begin{figure}
\else
\begin{figure*}
\fi
\begin{centering}
\ifshortver
\else
\includegraphics[scale=0.8]{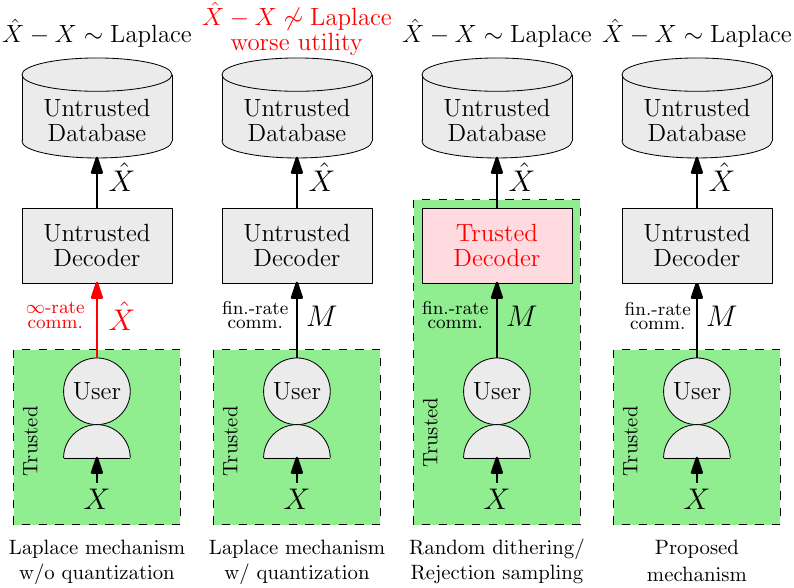}
\fi
\par\end{centering}
\caption{\label{fig:compare}Comparison between privacy mechanisms. There are 3 desirable properties: A: finite communication rate (description $M$ has a finite number of bits); B: Laplace-distributed noise $\hat{X}-X$ (the optimal utility among additive noise mechanisms for $d$-privacy~\cite{koufogiannis2015optimality}); and C: privacy guarantee against the decoder. Each existing scheme violates one of them. Laplace mechanism w/o quantization violates A, Laplace mechanism w/ quantization violates B, and existing random dithering and rejection sampling methods (e.g. \cite{harsha2010communication,sfrl_trans,hegazy2022randomized,flamich2023greedy,lang2023joint}) do not guarantee C. The proposed mechanism in this paper satisfies all A, B and C.
}
\ifshortver
\vspace{-10pt}
\else
\end{figure*}
\fi
\ifshortver
\end{figure}
\fi

This is a local differential privacy setting since neither the decoder nor the database is trusted.
We separate decoder and database privacy to allow a finer control granularity.
If we have a higher level of trust on the piece of software
at the server handling the decoding of packets (which only the administrator can access), compared to any party (e.g. third-party data analyst) that may access the
information in the database later, then we can impose a more stringent database privacy constraint than the decoder privacy constraint. 
\iflongver
Also, we can exploit quantization
noise to provide better database privacy than the decoder
privacy. 
Quantization noise is not only a necessary downside of quantization,
but can act as a privacy-preserving noise to improve database privacy.
\fi

\iflongver
A simple approach is to first apply a privacy mechanism (e.g. Laplace), and then
quantize. 
This is technically the approach adopted by many existing
systems, where a privacy mechanism (e.g. Laplace) is applied on
the data stored in machine precision (e.g. $32$-bit floating point),
and then quantized to a lower precision (e.g. $16$-bit floating point or $8$-bit integer)
before transmission. 
This can ensure both decoder and database privacy.
However, there are two downsides. First, the overall (Laplace plus
quantization) error will no longer be Laplace-distributed, and the utility will be degraded.
Second, the quantization noise here is unhelpful in improving database
privacy, and the privacy parameter for database is only the same
as that for the decoder.
\fi

Locally differentially private requantization~\cite{duchi2013local,xiong2016randomized} is a technique for finding the randomized quantization function mapping $X$ to a finite set of values of $\hat{X}$ with the best privacy guarantee at the decoder. The size of the support of $\hat{X}$ is bounded by communication constraint. This 
\ifshortver
\else
additional constraint on $\hat{X}$ 
\fi
may result in a degradation in the database privacy-utility tradeoff.
\ifshortver
\else
In \cite{chaudhuri2022privacy}, the random quantization function is designed according to the minimum variance unbiased criteria.
Also refer to \cite{youn2023randomized} for another related random quantization mechanism. In all the aforementioned mechanisms, the set of values of $\hat{X}$ is discrete and fixed.
\fi

If it is desirable to have $\hat{X}$ following a continuous distribution (e.g. the Laplace mechanism), we cannot fix the support of $\hat{X}$ to a finite set, but rather we should utilize the shared randomness $U$ to randomize the decoder as well.
This can be achieved by subtractive dithering 
\ifshortver
\else
\cite{gray1998quantization,gray1993dithered}. 
\fi
It was noted
in \cite{hegazy2022randomized} and \cite{lang2023joint} that subtractive
dithering, applied in a randomized manner, can produce a continuous nonuniform
error distribution. In particular, \cite{lang2023joint} studied the
use of randomized subtractive dithering to create privacy-preserving
noises such as multivariate $t$-distributed noises and Laplacian
noises. Channel simulation 
\ifshortver
\else
\cite{bennett2002entanglement,li2018universal,sfrl_trans}
\fi
and rejection sampling
\ifshortver
\else
\cite{harsha2010communication,havasi2019minimal,shah2022optimal,flamich2023greedy}
\fi
are techniques that can achieve the same goal. 
However, these techniques themselves can only
guarantee that $\hat{X}$ is distributed as if it is $X$ corrupted
by additive noises (guaranteeing database privacy), but generally cannot guarantee
\ifshortver
\else
that the quantized signal together with the dither signal (i.e., $(U,M)$)
would not leak information about $X$ (needed for decoder privacy), since one can deduce which quantization cell contains $X$ using the dither signal and the quantized signal.
This might not be desirable if the decoding component of the server is
not completely trusted. 
\fi
The privacy analyses in \cite{lang2023joint}
focuses on $\hat{X}$ instead of $(U,M)$, which is for database privacy, not decoder  
\ifshortver
\else
privacy (see Remark~\ref{rem:jopeq}).
\fi

We have seen that none of the existing mechanisms can realize the Laplace mechanism exactly with decoder privacy and a finite communication rate.
In this paper, we propose 
the first differential privacy mechanism realizing the Laplace mechanism that has a provable guarantee on decoder privacy and bound on communication cost.
In our mechanism, the recovered
data $\hat{X}$ is distributed exactly as if it is the input data
$X$ plus a Laplace noise. Therefore, it can be applied to any application
where the Laplace mechanism is 
\ifshortver
\else
used, without introducing any additional distortion to the Laplace-noise-perturbed data, while reducing the communication cost. 
While database privacy and utility are guaranteed by the Laplace
mechanism, we also prove that our mechanism attains decoder privacy. This 
is achieved by designing a mixture of dithered quanitizers with the correct shapes and weights that results in a desirable tradeoff between communication cost and decoder privacy.

Our mechanism is also applicable to centralized/global differential privacy \cite{Dwork06DP,Dwork14Book}, where a trusted curator holds the data $D$ of users, and answers queries from third party data analysts about the statistics of the data in the form of some query function $f(D)$. The curator may require all queries to be made via an API library, which performs the decoding of the packets sent from the curator to the analysts. The previously discussed ``database privacy'' now corresponds to the privacy of the data $D$ with respect to the sanitized statistics $\hat{X}$ output by the API (which is a noisy version of $f(D)$), which is the privacy guarantee against ``lawful'' analysts that only uses the API as intended. The ``decoder privacy'' now corresponds to the privacy of the data $D$ with respect to all information available to the API, which is the privacy guarantee against analysts who reverse-engineer the API to gain more information. While both privacy guarantees are important, it is natural to expect that most analysts are ``lawful'', and impose a more stringent privacy guarantee against them.
\fi

\iflongver
\begin{rem}
\label{rem:jopeq}
The privacy analysis in \cite{lang2023joint} focuses on $\hat{X}$ (i.e., database privacy) instead of $(U,M)$ (i.e., decoder privacy) because $\hat{X}$ is ``less distorted, and thus potentially more leaky'' compared to $M$. While this may be intuitively true, there are situations where $(U,M)$ leaks more information than $\hat{X}$ (e.g. in \cite{hegazy2022randomized} and the mechanism proposed in this paper), so focusing on $\hat{X}$ is insufficient to formally prove decoder privacy.

Also, it was noted in \cite{lang2023joint} that the JoPEQ mechanism can work with an infinitesimal privacy-preserving noise (the local randomness at the user), relying almost
solely on subtractive dithering (the shared randomness $U$). 
Some channel simulation and rejection sampling schemes \cite{harsha2010communication,sfrl_trans,hegazy2022randomized,flamich2023greedy} also do not utilize local randomness at the user, and rely solely on the shared randomness to create the desired noise.
We can show that schemes without local randomness cannot achieve decoder privacy (though they may achieve database privacy).
If there is no local randomness at the user, the description $M=g(X,U)$ is a function of the data $X$ and the shared randomness $U$. The conditional density can be computed as $f_{U,M|X}(u,m|x) = f_U(u) \mathbf{1}\{m=g(x,u)\}$, which changes discontinuously when $x$ changes, and hence decoder $d$-privacy cannot be achieved.
\end{rem}
\fi

\ifshortver
\else
\subsection*{Notations}
Scalar constants and random variables are denoted in lowercase (e.g. $x \in \mathbb{R}$) and uppercase (e.g, $X$), respectively. Vector constants and random variables are denoted in lowercase boldface (e.g. $\mathbf{x} \in \mathbb{R}^n$) and uppercase boldface (e.g, $\mathbf{X}$), respectively. 
\fi

\section{Preliminaries}

\subsection{Differential Privacy}
\iflongver
The notion of \emph{differential privacy} was  formalized in \cite{Dwork06DP}:

\medskip
\begin{defn}[$\epsilon$-differential privacy \cite{Dwork06DP,Dwork14Book,koufogiannis2015optimality}] \label{def:eDP}
    Let $\epsilon \geq 0$ be a given privacy budget, $\mathcal{X}$ be the set of possible private inputs (called \emph{databases}), $\mathcal{Y}$ be the set of possible outputs (called \emph{responses}), and $\Delta(\mathcal{Y})$ be the space of probability distributions over $\mathcal{Y}$. 
    Let $\mathcal{R}\subseteq \mathcal{X}^2$ be the set of pairs of adjacent databases.
    A randomized mechanism $\mathcal{A}:\mathcal{X} \rightarrow \Delta(\mathcal{Y})$ is $\epsilon$-\emph{differentially private} if for all $\mathcal{S} \subseteq \mathcal{Y}$ and for all $(x, x') \in \mathcal{R}$,
    \[
    \mathbb{P}(\mathcal{A}(x) \in \mathcal{S}) \leq e^{\epsilon}\mathbb{P}(\mathcal{A}(x') \in \mathcal{S}).
    \]
\end{defn}
\medskip

For centralized differential privacy \cite{Dwork06DP,Dwork14Book}, $\mathcal{R}$ is often taken to be the set of pairs of databases that differ in at most one record. For local differential privacy~\cite{Kasiviswanathan08LDP,duchi2013local}, we can take $\mathcal{R}=\mathcal{X}^2$. 

The adjacency condition in the original formulation \cite{Dwork06DP,Dwork14Book} depends how many records differ in two databases, i.e., the Hamming distance. Depending on the type of the element in $\mathcal{X}$, a suitable metric $d$ can be used to quantify the distance between $x$ and $x'$.
To broaden the scope of the differential privacy, Chatzikokolakis et al. \cite{Chatzikokolakis13Dprivacy} have introduced $d$-\emph{privacy} to relate the required privacy level to the distance between $x$ and $x'$ when $\mathcal{X}$ is a metric space equipped with a metric $d$. This is also known as geo-indistinguishability \cite{andres2013geo}.
\else
We review the notion of $d$-privacy or metric privacy \cite{Chatzikokolakis13Dprivacy,alvim2018local}, also known as geo-indistinguishability \cite{andres2013geo}.
\fi

\iflongver
\medskip
\else
\vspace{-6pt}
\fi
\begin{defn}[$\epsilon\cdot d$-privacy \cite{Chatzikokolakis13Dprivacy,andres2013geo,alvim2018local}] \label{def:dprivacy}
    Suppose $\mathcal{X}$ is a metric space equipped with a metric $d$, and $\mathcal{Y}$ is the set of possible outputs. Let $\epsilon>0$.
    A randomized mechanism $\mathcal{A}:\mathcal{X} \rightarrow \Delta(\mathcal{Y})$ 
    satisfies $\epsilon\cdot d$-\emph{privacy} if for every $\mathcal{S} \subseteq \mathcal{Y}$ and $x, x' \in \mathcal{X}$,
    \[
    \mathbb{P}(\mathcal{A}(x) \in \mathcal{S}) \leq e^{\epsilon \cdot d(x,x')}\mathbb{P}(\mathcal{A}(x') \in \mathcal{S}).
    \]
\end{defn}
\iflongver
\medskip
\else
\vspace{-6pt}
\fi

\iflongver
Under this formulation, the differential privacy constraint is viewed as a sensitivity constraint. 
Note that the standard $\epsilon$-differential privacy (Definition~\ref{def:eDP}) is a special case of $d$-privacy, by setting $d(x,x') = d_{H}(x,x')$ where $d_{H}(\cdot,\cdot)$ is the Hamming distance between $x$ and $x'$.
In \cite{koufogiannis2015optimality}, Koufogiannis and Pappas formulate a closely related notion of privacy for metric spaces called 
Lipchitz privacy.
The $d$-privacy is applied on a local differential privacy scenario in \cite{alvim2018local}.
\fi

A method to achieve $d$-privacy is the Laplace mechanism \cite{Dwork06DP,Dwork14Book,Chatzikokolakis13Dprivacy}.
Let $\mathrm{Laplace}(\mu,b)$ be the Laplace distribution with location $\mu$ and scale $b$, with a density function $\frac{1}{2b}e^{-\frac{|x - \mu|}{b}}$. The Laplace mechanism is defined as follows:
\iflongver
\medskip
\else
\vspace{-6pt}
\fi
\begin{defn}[Laplace mechanism \cite{Dwork06DP,Dwork14Book,Chatzikokolakis13Dprivacy}] \label{def:LaplaceM}
Let the data be $\mathbf{x}\in \mathcal{X}=\mathbb{R}^n$. The \emph{Laplace mechanism} is
\ifshortver
\else
\[
\mathcal{A}(\mathbf{x},\epsilon) := \mathbf{x} + \mathbf{Z},
\]
\fi
where $\mathbf{Z}_i \stackrel{iid}{\sim} \mathrm{Laplace}(0,1/\epsilon)$ for $i=1,\cdots,n$. 
\end{defn}
\iflongver
\medskip
\else
\vspace{-6pt}
\fi
It is shown in \cite{Chatzikokolakis13Dprivacy} that the Laplace mechanism satisfies $\epsilon \cdot d$-privacy, where $d(\mathbf{x},\mathbf{x}')=\Vert \mathbf{x}-\mathbf{x}'\Vert_1$ is the $\ell_1$-metric.

\iflongver
The Laplace mechanism can be applied to centralized differential privacy as well \cite{Dwork06DP,Dwork14Book}. For a query function $f : \mathcal{X} \to \mathbb{R}^n$, applying the Laplace mechanism $\mathcal{A}(f(x),\epsilon / \Delta f )$ on $f(x)$, we can achieve $\epsilon$-differential privacy, where $\Delta f$ is the $\ell_1$-\emph{sensitivity} of $f$ defined as:
\[
\Delta f:= \max_{x, x' \in \mathcal{X}: (x,x') \in \mathcal{R}}  \Vert f(x)-f(x')\Vert_1 
\]
\fi

\iflongver
\medskip{}
\fi

\subsection{Subtractive Dithering}
We review the subtractively dithered quantization 
\ifshortver
\else
\cite{gray1998quantization,gray1993dithered}. 
\fi
The encoder and decoder share a random variable $U \sim \mathrm{Unif}(-\frac{1}{2},\frac{1}{2})$. Fix a quantization step size $\delta>0$. To encode $x \in \mathbb{R}$, the encoder gives $M=E_\qs(x,U):=\mathrm{round}(x/\delta - U)$, where $\mathrm{round}(y)$ is the closest integer to $y$. The decoder uses $(M,U)$ to recover $\hat{X}=D_\qs(M,U):=\delta(M+U)$. We have $\hat{X}-x \sim \mathrm{Unif}(-\delta/2,\delta/2)$ for every $x\in \mathbb{R}$. The behavior of the quantization is the same as applying an additive uniformly distributed noise independent of the input.

\ifshortver
\else
Subtractive dithering can be generalized to vector lattice quantization~\cite{kirac1996results,zamir2014,shlezinger2020uveqfed}, though this is beyond the scope of this paper.
\fi

\iflongver
\medskip{}
\fi

\section{Problem Formulation}
\iflongver
We present the setting in this paper as a local differential privacy scenario. 
\fi
The user holds the data $x \in \mathcal{X}$, where $\mathcal{X}$ is the set of possible data. The user and the decoder have a shared randomness $U$. We allow unlimited shared randomness and allow $U$ to follow a continuous distribution. This is the same assumption as subtractive dithering \cite{gray1998quantization,hegazy2022randomized,lang2023joint} (where $U$ is the uniform dither) and some channel simulation schemes \cite{harsha2010communication,sfrl_trans}, 
\ifshortver
\else
and is a realistic assumption since most practical pseudorandom number generator can be initialized by a small seed (e.g. $64$ bits)
, and the user and the decoder only need to share a small
seed in order to provide a practically unlimited stream of shared random numbers for all the communication tasks that arise later.
\fi

The user will transmit a description $M=E(x,U,V) \in \mathcal{M}$ to the decoder, where $E$ is an encoding function, and $V$ is the (unlimited) local randomness at the user. The description is required to be in a discrete space $\mathcal{M}$, which is often taken to be either $\mathbb{Z}$ or the set of variable-length bit sequences $\{0,1\}^*$. 
\iflongver
Note that one can encode $M\in \mathbb{Z}$ into a variable-length bit sequence using codes such as the signed Elias delta code~\cite{elias1975universal}. 
\fi

After observing $M$ and $U$, the decoder will recover $\hat{X} = D(M,U)$, where $D$ is a decoding function. Our goal is to have $\hat{X}$ following a prescribed conditional distribution given $x$ that ensures privacy. We will focus on the situation where $x = \mathbf{x} \in \mathbb{R}^n$ is a vector, and $\hat{X} = \hat{\mathbf{X}} \in \mathbb{R}^n$ where the components of the error $\hat{\mathbf{X}}_i-\mathbf{x}_i$ (for $i = 1,\ldots,n$) are required to be i.i.d. following $\mathrm{Laplace}(0,1/\epsilon)$, the Laplace distribution centered at $0$ with scale $1/\epsilon$, with density function $e^{-\epsilon |x|}\cdot \epsilon/2$, for every value of $\mathbf{x} \in \mathbb{R}^n$. This ensures that the outward behavior of the whole mechanism is the same as that of the Laplace mechanism \cite{Dwork06DP,Dwork14Book}, except that we now only require the transmission of $M$ which can be encoded into a finite number of bits. In particular, the utility of $\hat{\mathbf{X}}$ is the same as that of the Laplace mechanism.

We say that the mechanism achieves a $\epsilon_1 \cdot d$\emph{-privacy against the database}, if the overall mechanism mapping $\mathbf{x}$ to $\hat{\mathbf{X}}$ satisfies $\epsilon_1 \cdot d$-privacy in Definition~\ref{def:dprivacy}, where $d(\mathbf{x},\mathbf{x}')=\Vert \mathbf{x}-\mathbf{x}'\Vert_1$ is the $\ell_1$ metric. 
\ifshortver
\else
This is the privacy of the data against any party who can only access $\hat{\mathbf{X}}$ (e.g. the database that stores the noise corrupted data of users). If we can indeed realize the aforementioned Laplace mechanism, then we have a $\epsilon \cdot d$-privacy against the database.
\fi

We say that the mechanism achieves a $\epsilon_2 \cdot d$\emph{-privacy against the decoder} if the random mapping from $\mathbf{x}$ to $(U,M)$ satisfies $\epsilon_2 \cdot d$-privacy where $d$ is the $\ell_1$ metric. This is the privacy of the data against the decoder who can access $(U,M)$. 
\iflongver
This is important when the decoder is not completely trusted. Refer to the introduction for a discussion. 
\fi
The aforementioned Laplace mechanism requirement is not sufficient to guarantee privacy against the decoder. 
\iflongver
For example, while the randomized dithered quantization scheme in~\cite{hegazy2022randomized} can realize the Laplace mechanism, it leaks too much information when the random quantization step size happens to be small, which exposes a precise approximation of the data. This will be the main difficulty in the design and analysis of our new mechanism.
\fi

\ifshortver
\else
We remark that our mechanism applies to centralized differential privacy as well (as discussed in the introduction), where the user becomes the data curator, and the decoder becomes the data analyst.
\fi

\iflongver
\medskip{}
\fi

\section{Quantized Piecewise Linear Mechanism}

We first consider the case where $x \in \mathbb{R}$ is a scalar, and the desired distribution of the error $\hat{X}-x$ is not Laplace, but the following piecewise linear probability density function
\begin{equation}
f(x):=\frac{1}{\qs}\sum_{k\in\mathbb{Z}}a_{k}\mathrm{tri}\left(\frac{x}{\qs}-k\right),\label{eq:piecewise_linear}
\end{equation}
where $\delta>0$ is the step size, $\mathrm{tri}(x)=\max\{1-|x|,\,0\}$ is the triangle function,
and $a_{k}\ge0$ with $\sum_{k}a_{k}=1$. Note that $f$ is the linear interpolation of the points $(k\delta, a_k/\delta)$ for $k\in \mathbb{Z}$. Since $f$ is the convolution between $\mathrm{rect}(x/\delta)/\delta$ and the step function
\ifshortver
\else
\[
\tilde{f}(x):=\frac{1}{\qs}\sum_{k\in\mathbb{Z}}a_{k}\mathrm{rect}\left(\frac{x}{\qs}-k\right),
\]
\fi
where $\mathrm{rect}(x)=\mathbf{1}\{|x|\le1/2\}$ is the rectangular
function (the probability density function of $\mathrm{Unif}(-1/2,1/2)$), the user can simply apply a local noise with density $\tilde{f}$ to $x$ before using a subtractively dithered quantizer with step size $\delta$ to achieve an overall error distribution $f$. 

We now define the \emph{quantized piecewise linear mechanism} on $f$. Assuming that the local randomness is $(V,W)$, where $V \sim \mathrm{Unif}(0,1)$ independent of $W \sim \mathrm{Unif}(-1/2,1/2)$, the encoding and decoding functions can be constructed as
\begin{align}
&E_{f,\qs}(x,u,(v,w)) \nonumber \\
&:=\mathrm{round}\Big(\frac{x}{\delta} + \min\Big\{\! k\in\mathbb{Z}:\!\!\!\sum_{i=-\infty}^{k}\!\! a_i \ge v\Big\} +w - u\Big) ,\label{eq:piecewise_enc}
\end{align}
\ifshortver
\else
\[
D_{\qs}(m,u) :=\qs(m+u).
\]
\fi
Note that $\min\{ k\in\mathbb{Z}:\,\sum_{i=-\infty}^{k}a_i \ge V\}$ is simply a method to generate a random integer with probabilities $a_i$ using $V \sim \mathrm{Unif}(0,1)$. 
\iflongver
Hence $\min\{ \cdots \} + W$ follows the distribution $\tilde{f}$ scaled by $1/\delta$. 
\fi

\ifshortver
\else
As an intermediate step, we prove the privacy of the quantized piecewise linear mechanism.

\medskip
\begin{lem}
\label{lem:piecewise_lip}The quantized piecewise linear mechanism achieves $\epsilon \cdot d$-privacy against the decoder and the database, where $d(x,x')=|x-x'|$, and
\[
\epsilon = \max_{i,j\in\mathbb{Z}:\,|i-j|=1}\frac{1}{\qs}\Big(\frac{a_{j}}{a_{i}}-1\Big).
\]
\end{lem}
\begin{proof}
We first find the Lipschitz constant of $\ln f(x)$ for $0\le x \le \delta$. For $0\le x \le \delta$,
\begin{align*}
\Big|\frac{\mathrm{d}}{\mathrm{d}x}\ln f(x)\Big| &= \Big|\frac{\mathrm{d}}{\mathrm{d}x}\ln \Big( \frac{a_0}{\delta} + x\frac{a_1-a_0}{\delta^2} \Big) \Big|\\
&= \Big|\frac{a_1-a_0}{a_0 \delta + x(a_1-a_0)}\Big|\\
& \le \max\Big\{ \Big|\frac{a_1-a_0}{a_0 \delta }\Big| ,\, \Big|\frac{a_1-a_0}{a_0 \delta + \delta(a_1-a_0)}\Big| \Big\} \\
& = \max\Big\{ \frac{1}{\qs}\Big(\frac{a_{1}}{a_{0}}-1\Big) ,\, \frac{1}{\qs}\Big(\frac{a_{0}}{a_{1}}-1\Big) \Big\}.
\end{align*}
The other intervals for $x$ are similar. Hence the Lipschitz constant of $\ln f(x)$ is bounded by $\epsilon$. This gives the $\epsilon \cdot d$-privacy against the database. For privacy against the decoder, the conditional density function of $(U,M)$ given $x$ is
\begin{align*}
&f_{U,M|x}(u,m|x) \\
&= p_{M|x,U}(m|x,u)f_{U|x}(u|x) \\
&= p_{M|x,U}(m|x,u) \\
&\stackrel{(a)}{=} \mathbb{P}\Big( \Big|\frac{x}{\delta} + \min\Big\{ k\in\mathbb{Z}:\!\!\! \sum_{i=-\infty}^{k}\!\!\! a_i \ge V \!\Big\} +W - u -m\Big| \le \frac{1}{2} \Big) \\
&\stackrel{(b)}{=} \int_{m+u-x/\qs-1/2}^{m+u-x/\qs+1/2} \delta \tilde{f}(\delta y) \mathrm{d}y \\
&= \int_{\delta(m+u-1/2)-x}^{\delta(m+u+1/2)-x} \tilde{f}(y) \mathrm{d}y \\
&\stackrel{(c)}{=} \delta  f(\delta(m+u)-x),
\end{align*}
where (a) is by \eqref{eq:piecewise_linear}, (b) is because $\min\{ \cdots \} + W$ follows the distribution $\tilde{f}$ scaled by $1/\delta$, and (c) is because $f$ is the convolution of $\tilde{f}$ and $\mathrm{rect}(x/\delta)/\delta$. Hence $\ln f_{U,M|x}(u,m|x)$ is $\epsilon$-Lipschitz in $x$.
\end{proof}
\medskip

We remark that the general approach of applying local noise before quantization is common in conventional (non-subtractive) dithering
\ifshortver
\cite{gray1998quantization},
\else
\cite{gray1998quantization,gray1993dithered},
\fi
and has also been used in \cite{lang2023joint}. The novelty in this paper lies in the design of the local noise distribution, and the randomization procedure in the next section.

\fi

\iflongver
\medskip{}
\fi

\section{Proposed Privacy Mechanism}

Previously, we have discussed a mechanism for a piecewise linear error distribution.
Unfortunately, the Laplace distribution is not
piecewise linear. Nevertheless, we can still construct a mechanism by
decomposing the Laplace distribution into a mixture of piecewise linear
distributions for different step sizes $\qs$'s, and apply the quantized piecewise linear mechanism on a randomly chosen component of the mixture.
\ifshortver
\else
Basically, we will find a sequence of distributions $f_0,f_{1},f_{2},\ldots$,
where $f_{i}$ is piecewise linear with step size $\qs_i$ in the form \eqref{eq:piecewise_linear}, and $\sum_{i}\gamma_{i}f_{i}(x)$
is the Laplace distribution for some weights $\gamma_{i}\ge0$, $\sum_{i}\gamma_{i}=1$.
We will first generate a shared randomness $T$ with $\mathbb{P}(T=t)=\gamma_{t}$,
and then apply the quantized piecewise linear mechanism on $f_{T}$. Since both
the encoder and the decoder knows $T$, they know $\qs_{T}$ and
$f_{T}$, and can perform the encoding and decoding of the quantized piecewise linear mechanism on $f_{T}$. The resultant error distribution will be the
Laplace distribution.
\fi

Define the \emph{piecewise linear Laplace distribution}
\begin{align}
g_{\qs}(x) :=\frac{1}{c_{\qs}}\sum_{k\in\mathbb{Z}}e^{-|k\qs|}\mathrm{tri}\left(\frac{x}{\qs}-k\right), \label{eq:g_delta}
\end{align}
where $\qs>0$, and
\ifshortver
\else
\begin{align*}
c_{\qs} & :=\qs\sum_{k\in\mathbb{Z}}e^{-|k\qs|}=\qs\frac{1+e^{-\qs}}{1-e^{-\qs}}.
\end{align*}
\fi
This is in the form \eqref{eq:piecewise_linear} by taking $a_k=\delta e^{-|k\qs|}/c_{\qs}$. This is an approximate version of the Laplace distribution that is piecewise linear, and tends to the Laplace distribution as $\qs \to 0$.

We now define the \emph{dyadic quantized Laplace (DQL) mechanism}. It has two parameters:  the \emph{database privacy parameter}
$\epsilon>0$, and the \emph{decoder privacy relaxation parameter} $\ell > 1$. We first consider the scalar case $x \in \mathbb{R}$. Let $\delta_0 > 0$ be the solution to 
the equation $e^{\delta_{0}}=\delta_{0}\ell+1$, and $\delta_{t}:=2^{-t}\delta_{0}$ for $t \in \mathbb{Z}_{>0}$. The shared randomness are $U\sim\mathrm{Unif}(-1/2,1/2)$
and $T\in\mathbb{Z}_{\ge0}$ which is a random variable with a cdf $F_{T}(t)=\mathbb{P}(T\le t)$,
\begin{align}
F_{T}(t) &:=\!\!\prod_{i=t+1}^{\infty}\frac{4-4(\delta_{i}\ell+1)e^{-\delta_{i}}}{(1+e^{-\delta_{i}})^{2}\big(2(1+e^{-2\delta_{i}})^{-1}\! -\delta_{i}\ell-1\big)} \label{eq:F_T}
\end{align}
for $t=-1,0,1,\ldots$. Given $x,T,U$ and the local randomness
$V\sim\mathrm{Unif}(0,1)$, $W\sim\mathrm{Unif}(-1/2,1/2)$, the encoder produces the description
\begin{align*}
M &= E_{\mathrm{DQL}}(x,(T,\! U),(V,\! W))  := E_{f_{T},\,\qs_T}(\epsilon x,U,(V,\! W)),
\end{align*}
where 
\begin{align}
f_{t}(x)=\frac{F_{T}(t)g_{\qs_t}(x)-F_{T}(t-1)g_{\qs_{t-1}}(x)}{F_{T}(t)-F_{T}(t-1)} \label{eq:f_T}
\end{align}
is a piecewise linear density function in the form \eqref{eq:piecewise_linear} with step size $\qs_t$, and $E_{f,\qs}$
is given in \eqref{eq:piecewise_enc}. It will be proved in Theorem~\ref{THM:PRIVACY} that $f_t(x)\ge 0$. Given $M,T,U$, the decoder recovers
\begin{align}
\hat{X}=D_{\mathrm{DQL}}(M, (T,U)) := \qs_T (M+U)/\epsilon. \label{eq:dyadic_dec}
\end{align}
Refer to Fig.~\ref{fig:lerp_1} for an explanation. The vector case, with input $\mathbf{x}\in \mathbb{R}^n$ and output $\hat{\mathbf{X}}\in \mathbb{R}^n$, can be obtained by applying this mechanism on each entry separately. 
\ifshortver
\else
Consider $\mathbf{x} \in \mathbb{R}^n$. The shared randomness are $\mathbf{U} \in \mathbb{R}^n$, $\mathbf{U} \stackrel{iid}{\sim}\mathrm{Unif}(-1/2,1/2)$
and $\mathbf{T} \in \mathbb{Z}_{\ge 0}^n$, $\mathbf{T} \stackrel{iid}{\sim} p_T$ (where $p_T$ is the distribution defined in the scalar case). Given $\mathbf{x},\mathbf{T},\mathbf{U}$ and the local randomness
$\mathbf{V}\stackrel{iid}{\sim}\mathrm{Unif}(0,1)$, $\mathbf{W}\stackrel{iid}{\sim}\mathrm{Unif}(-1/2,1/2)$, the encoder produces the description
\begin{align*}
\mathbf{M} &= E_{\mathrm{DQL}}(\mathbf{x},(\mathbf{T},\mathbf{U}),(\mathbf{V},\mathbf{W})) \\
&:= (E_{\mathrm{DQL}}(\mathbf{x}_i,(\mathbf{T}_i,\mathbf{U}_i),(\mathbf{V}_i,\mathbf{W}_i)))_{i =1,\ldots,n}.
\end{align*}
Given $\mathbf{M},\mathbf{T},\mathbf{U}$, the decoder recovers
\[
\hat{\mathbf{X}}=D_{\mathrm{DQL}}(\mathbf{M}, (\mathbf{T},\mathbf{U})) := (D_{\mathrm{DQL}}(\mathbf{M}_i, (\mathbf{T}_i,\mathbf{U}_i)))_{i =1,\ldots,n}.
\]
\fi

\ifshortver
\else
\begin{figure*}
\begin{centering}
\includegraphics[scale=0.3]{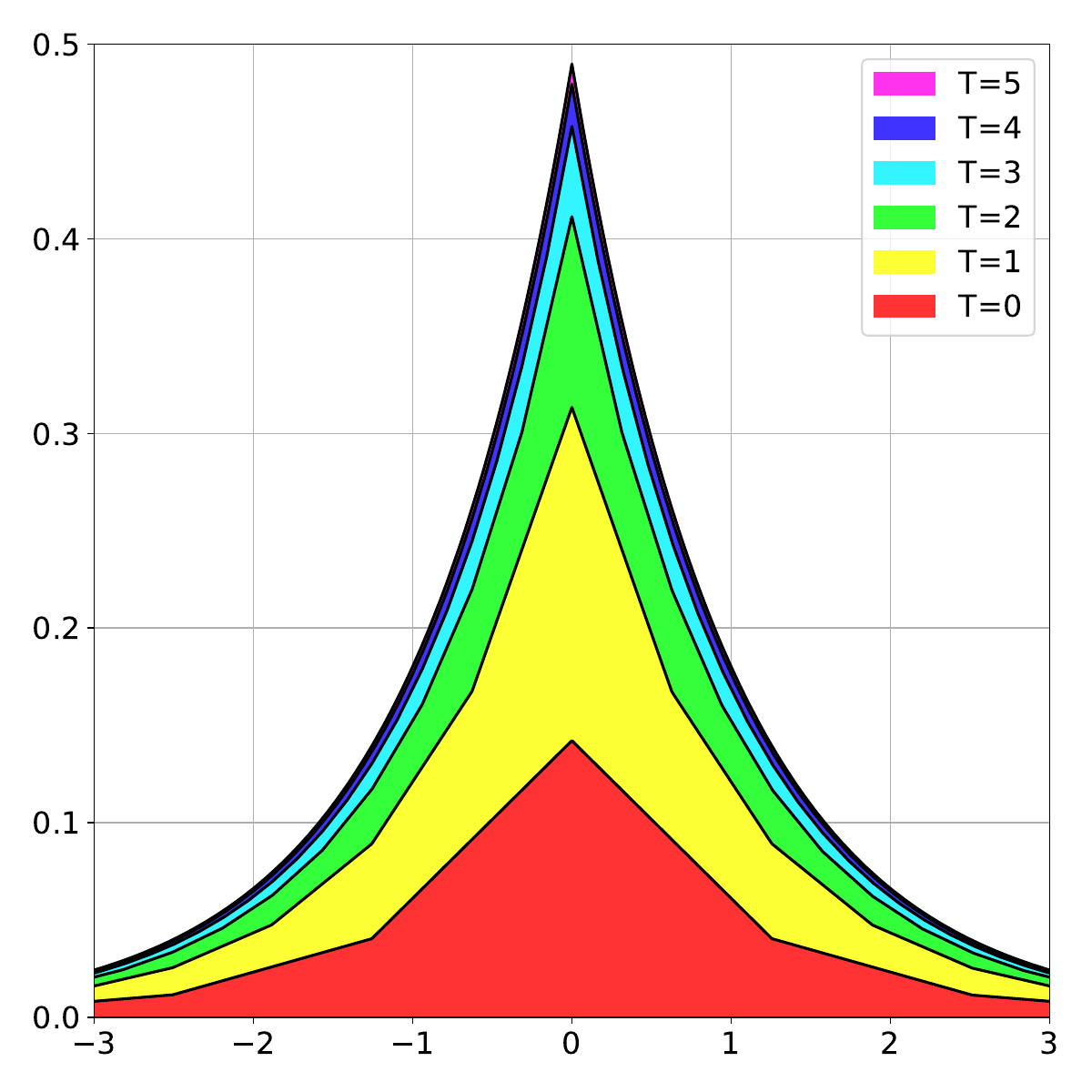}\includegraphics[scale=0.3]{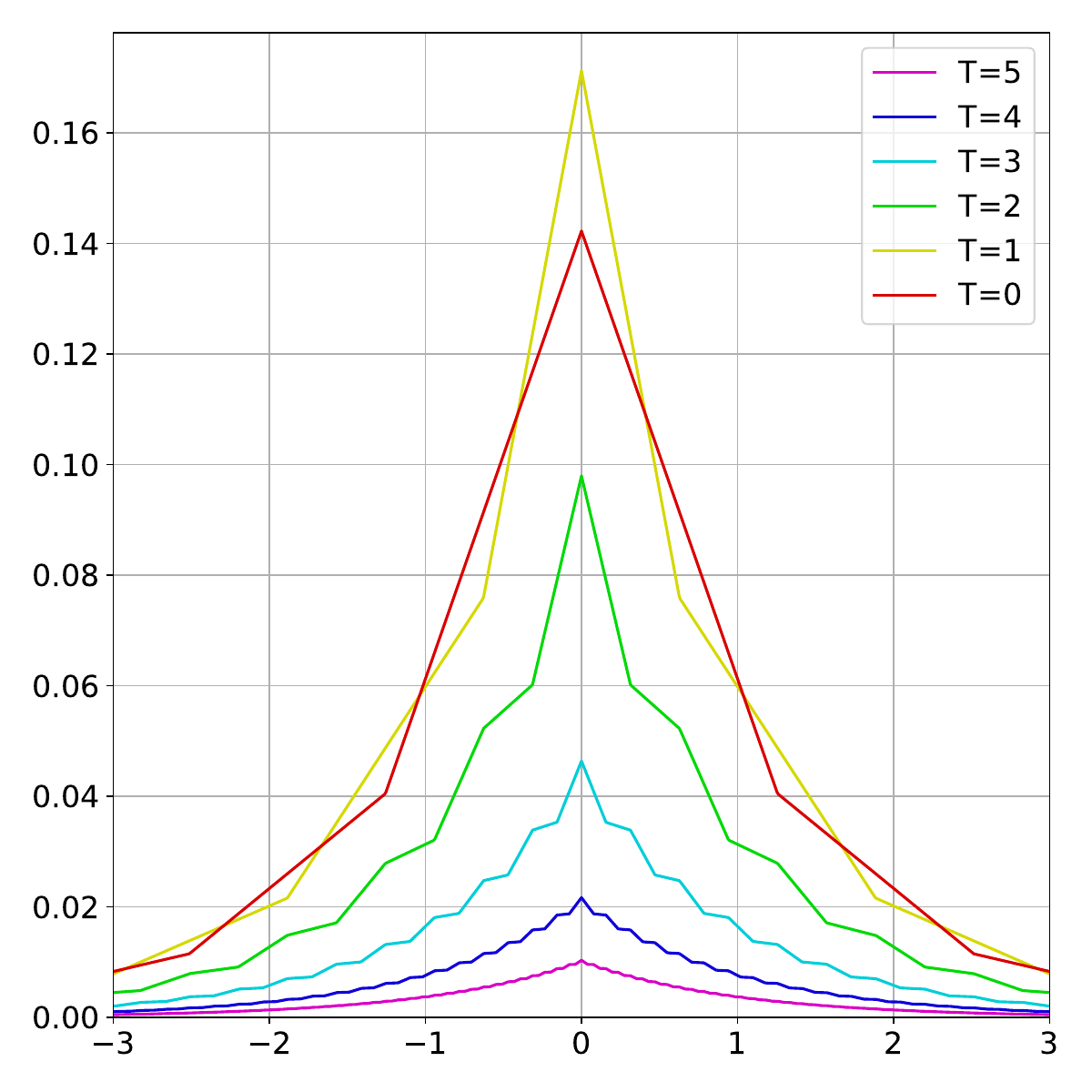}\includegraphics[scale=0.3]{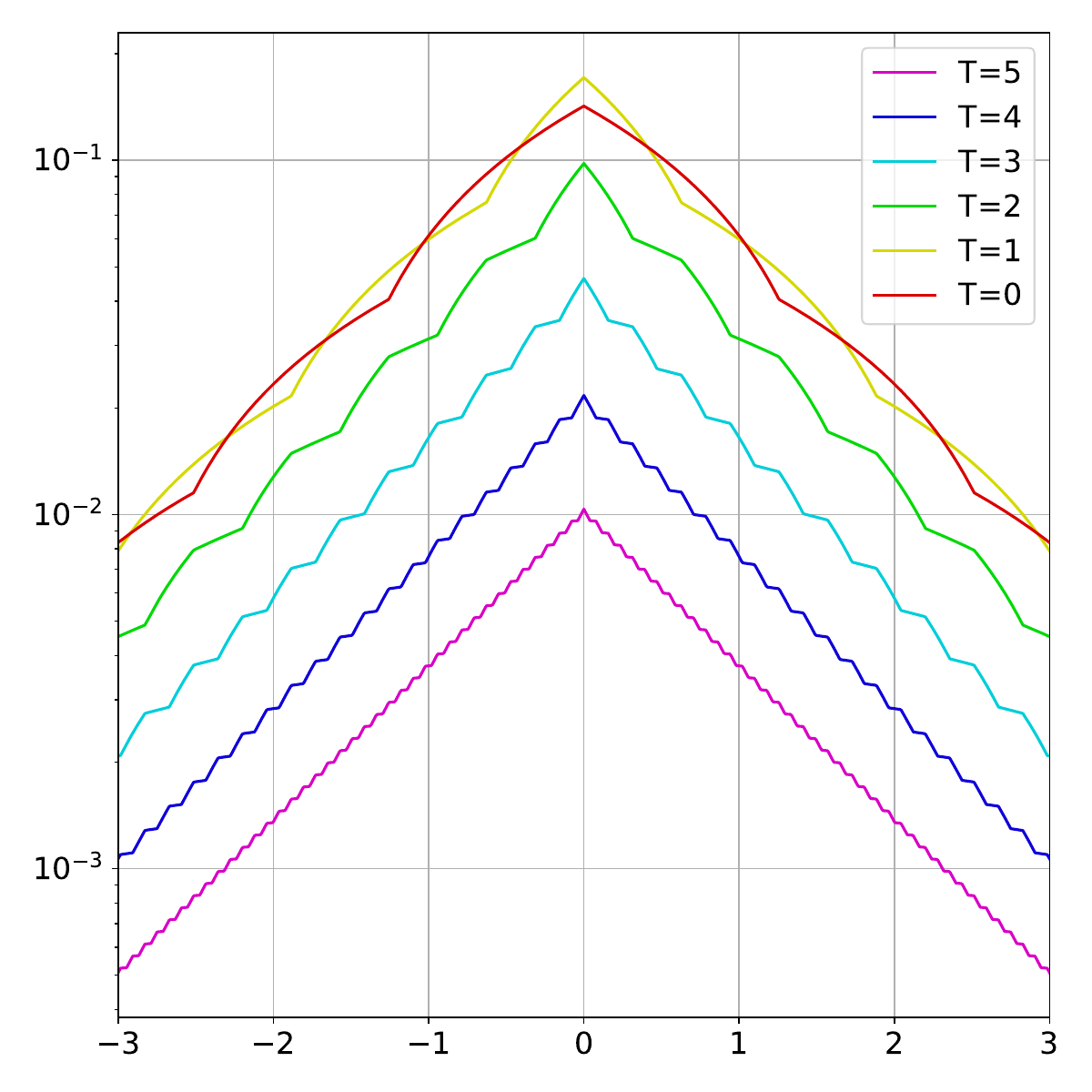}
\par\end{centering}
\caption{\label{fig:lerp_1}Left: Plot of $F_{T}(t)g_{2^{-t}\qs_0}(x)$ for
different values of $t$, where $b=1$, $\ell=2$. Note that
$p_{T}(1)f_{1}(x)$ is the red area below $F_{T}(1)g_{2^{-1}\qs_0}(x)$,
and $p_{T}(2)f_{2}(x)$ is the yellow area between the curves $F_{T}(1)g_{2^{-1}\qs_0}(x)$
and $F_{T}(2)g_{2^{-2}\qs_0}(x)$, etc. The dyadic quantized
Laplace mechanism can be regarded as first picking an area (red, yellow,
green, etc.) with probability equal to their areas, and then
applying the quantized piecewise linear mechanism over that area. These
areas add up to the Laplace distribution. Middle: Plot of $p_{T}(t)f_{t}(x)$ for different values
of $t$. These piecewise-linear curves add up to the Laplace distribution.
Right: Plot of $p_{T}(t)f_{t}(x)$ in log scale for different values
of $t$, showing that $\log p_{T}(t)f_{t}(x)$ is Lipschitz continuous.}
\end{figure*}
\fi


\iflongver
It is straightforward to check that $\hat{\mathbf{X}}-\mathbf{x}$ has the desired Laplace distribution. The proof is omitted.
\fi
\begin{prop}
The DQL mechanism has 
\ifshortver
\else
an overall quantization 
\fi
noise $\hat{\mathbf{X}}_i-\mathbf{x}_i\sim\mathrm{Laplace}(0,1/\epsilon)$ i.i.d. for $i=1,\ldots,n$, for every 
\iflongver
given input 
\fi
$\mathbf{x} \in \mathbb{R}^n$.
\end{prop}

\ifshortver
\else
Since we have $\hat{\mathbf{X}}_i-\mathbf{x}_i \stackrel{iid}{\sim} \mathrm{Laplace}(0,1/\epsilon)$, we have $\epsilon\cdot d$-privacy against the database. In the following theorem, we prove that it also provides $\ell\epsilon \cdot d$-privacy against the decoder.  
The proof is given in Appendix~\ref{subsec:privacy_pf}.
\fi

\iflongver
\medskip{}
\fi
\begin{thm} 
\label{THM:PRIVACY} For $\mathbf{x} \in \mathbb{R}^n$, the DQL mechanism satisfies
$\ell \epsilon \cdot d$-privacy against the decoder, and $\epsilon \cdot d$-privacy against the database, where the metric is $d(\mathbf{x},\mathbf{x}')=\Vert \mathbf{x} - \mathbf{x}' \Vert_1$.
\end{thm}
\iflongver
\medskip{}
\fi

To encode the description $\mathbf{M}\in\mathbb{Z}^{n}$ into bits, we can use any code over integers. 
\ifshortver
\else
In particular, the following result bounds the communication cost if we apply the signed Elias delta code~\cite{elias1975universal}.

\begin{thm}
For vector data $\mathbf{x}\in\mathbb{R}^{n}$, applying signed Elias
delta code on each entry of the description $\mathbf{M}\in\mathbb{Z}^{n}$,
the expected description length (in bits) is upper-bounded by
\[
nL\left(\ln\left(\frac{2\epsilon\Vert\mathbf{x}\Vert_{1}}{n}+\frac{9}{8}\ln(2\ell\ln\ell+1)+2\right)+\mathbb{E}[-\ln\delta_{T}]\right),
\]
where
\[
L(z):=z\log_{2}e+2\log_{2}(z\log_{2}e+1)+1.
\]
Numerical evaluation gives the following upper bound on the expected
length:
\[
nL\!\bigg(\!\ln\!\Big(\frac{2\epsilon\Vert\mathbf{x}\Vert_{1}}{n}+\frac{9}{8}\ln(2\ell\ln\ell+1)+2\Big)+\ln\!\Big(\frac{e}{\ell\! -\! 1}+1\Big)-\frac{1}{2}\bigg).
\]
\end{thm}
\begin{proof}
We first prove the scalar case. 
Consider the function
$\ell^{2}-2\ell\ln\ell-1$. We have $\frac{\mathrm{d}}{\mathrm{d}\ell}(\ell^{2}-2\ell\ln\ell-1)=2(\ell-\log\ell-1)\ge0$
for $\ell\ge1$. Hence $\ell^{2}-2\ell\ln\ell-1\ge0$ for $\ell\ge1$,
$\ell^{2}\ge2\ell\ln\ell+1$, and
\[
e^{\ln(2\ell\ln\ell+1)}=2\ell\ln\ell+1\ge\ell\ln(2\ell\ln\ell+1)+1.
\]
Since $\delta_{0}$ satisfies $e^{\delta_{0}}=\delta_{0}\ell+1$,
we have
\[
\delta_{0}\le\ln(2\ell\ln\ell+1).
\]
By~\eqref{eq:dyadic_dec}, we have
\[
M=\frac{\epsilon\hat{X}}{\delta_{T}}-U.
\]
Hence,
\begin{align*}
 & \mathbb{E}\left[\ln(2|M-1/4|+1/2)\right]\\
 & =\mathbb{E}\left[\ln\left(2\left|\frac{\epsilon\hat{X}}{\delta_{T}}-U-\frac{1}{4}\right|+\frac{1}{2}\right)\right]\\
 & \le\mathbb{E}\left[\ln\left(\frac{2\epsilon |\hat{X}|}{\delta_{T}}+\frac{2|U+1/4|+1/2}{\delta_{T}/\delta_{0}}\right)\right]\\
 & =\mathbb{E}\left[\ln\left(2\epsilon |\hat{X}|+2\delta_{0}|U+1/4|+\delta_{0}/2\right)\right]+\mathbb{E}[-\ln\delta_{T}]\\
 & \le\ln\left(2\epsilon\mathbb{E}[|\hat{X}|]+2\delta_{0}\mathbb{E}[|U+1/4|]+\delta_{0}/2\right)+\mathbb{E}[-\ln\delta_{T}]\\
 & \le\ln\left(2\epsilon|x|+2\epsilon\mathbb{E}[|\hat{X}-x|]+\left(\frac{5}{8}+\frac{1}{2}\right)\delta_{0}\right)+\mathbb{E}[-\ln\delta_{T}]\\
 & =\ln\left(2\epsilon|x|+\frac{9}{8}\delta_{0}+2\right)+\mathbb{E}[-\ln\delta_{T}]\\
 & \le\ln\left(2\epsilon|x|+\frac{9}{8}\ln(2\ell\ln\ell+1)+2\right)+\mathbb{E}[-\ln\delta_{T}].
\end{align*}
Since the length of the unsigned Elias delta coding of $k\ge1$ is
upper-bounded by $L(\ln k)$~\cite{elias1975universal}, the expected length of the
signed Elias delta coding of $M$ (which is the unsigned Elias delta
coding of $2|M-1/4|+1/2$) is upper-bounded by
\[
L\left(\ln\left(2\epsilon|x|+\frac{9}{8}\ln(2\ell\ln\ell+1)+2\right)+\mathbb{E}[-\ln\delta_{T}]\right).
\]
Note that $\mathbb{E}[-\ln\delta_{T}]=\mathbb{E}[T]\ln2-\ln\delta_{0}$
is a function of $\ell$. We can check numerically via a plot that
\begin{align*}
\mathbb{E}[-\ln\delta_{T}] & \le\ln\left(\frac{e}{\ell-1}+1\right)-\frac{1}{2}.
\end{align*}
Hence, the expected length is upper-bounded by
\[
L\bigg(\ln\Big(2\epsilon|x|+\frac{9}{8}\ln(2\ell\ln\ell+1)+2\Big)+\ln\Big(\frac{e}{\ell-1}+1\Big)-\frac{1}{2}\bigg).
\]
The vector case follows from applying this bound on each entry of $\mathbf{x}$, and the concavity of $L(z)$.
\end{proof}
\medskip

Hence, for large $\epsilon\Vert\mathbf{x}\Vert_1$ and a bounded $\ell$, the expected length is $n\log_2(\epsilon\Vert\mathbf{x}\Vert_1/n) + O(n\log\log(\epsilon\Vert\mathbf{x}\Vert_1/n) + n \log(1/(\ell - 1)))$.

\fi

\iflongver
\medskip
\fi

\iflongver
\section{Algorithms}

The encoding and decoding algorithms are given in Algorithms \ref{alg:enc} and \ref{alg:dec}. We assume that the encoder has a pseudorandom number generator (PRNG) $\mathcal{P}$, and the decoder has a PRNG $\mathcal{P}'$, where the two PRNGs are initialized with the same seed (which was communicated ahead of time). The two PRNGs are kept synchronized, so we have $\mathcal{P}'=\mathcal{P}$, and we can invoke them to obtain shared randomness. These PRNGs can be reused for other tasks that require shared randomness. The encoder also has another PRNG initialized using a local random seed unknown to the decoder. 

The $g_\delta$ in \eqref{eq:g_delta} is a piecewise linear version of the Laplace distribution (see Figure~\ref{fig:lerp_1}, left figure). The vertices of $g_\delta$ can be broken into two sequences: those with nonnegative $x$-coordinates with exponentially decreasing $y$-coordinates, and those with negative $x$-coordinates with exponentially increasing $y$-coordinates. We can represent a random variable following $g_\delta$ as a sum of a $\mathrm{Unif}(-\delta/2,\delta/2)$ random variable and a mixture of two (appropriately scaled) geometric random variables, one representing the nonnegative $x$ vertices, and one representing the negative $x$ vertices.
Similarly, the vertices $f_t$ in \eqref{eq:f_T} can be divided into four sequences: nonnegative even $x$, nonnegative odd $x$, negative even $x$, and negative odd $x$ (see Figure~\ref{fig:lerp_1}, e.g., green region and lines). Therefore, the quantized piecewise linear mechanism can be carried out by sampling from a mixture of four geometric distributions, as given in Algorithm \ref{alg:enc}.
\fi

\begin{algorithm}[h]
\ifshortver
\else
\textbf{$\;\;\;\;$Input:} data $x\in\mathbb{R}$, parameters $\epsilon > 0$, $\ell > 1$, PRNG $\mathcal{P}$

\textbf{$\;\;\;\;$Output:} description $M \in \mathbb{Z}$

\smallskip{}
\fi

\begin{algorithmic}[1]

\State{Let $\delta_{0}>0$ be the solution to $e^{\delta_{0}}=\delta_{0}\ell+1$}

\State{Use PRNG $\mathcal{P}$ to sample $T\in\mathbb{Z}_{\ge0}$ with cdf given by \eqref{eq:F_T}}

\State{Use PRNG $\mathcal{P}$ to sample $U\sim\mathrm{Unif}(-1/2,1/2)$}

\State{Switch to another PRNG initialized with local seed}

\ifshortver
\else

\State{$\qs\leftarrow2^{-T}\qs_0$}

\State{$c_{0}\leftarrow\qs\frac{1+e^{-\qs}}{1-e^{-\qs}}$}

\State{$c_{1}\leftarrow2\qs\frac{1+e^{-2\qs}}{1-e^{-2\qs}}$}

\fi

\State{$r\leftarrow F_{T}(T-1)/F_{T}(T)$ (where $F_{T}$ is given in \eqref{eq:F_T})}

\State{$(M_{0},Z)\leftarrow\begin{cases}
(0,2) & \text{w.p.}\propto\frac{1}{c_{0}}-\frac{r}{c_{1}}\\
(-2,-2) & \text{w.p.}\propto(\frac{1}{c_{0}}-\frac{r}{c_{1}})e^{-2\qs}\\
(1,2) & \text{w.p.}\propto\frac{e^{-\qs}}{c_{0}}-\frac{r(1+e^{-2\qs})}{2c_{1}}\\
(-1,-2) & \text{w.p.}\propto\frac{e^{-\qs}}{c_{0}}-\frac{r(1+e^{-2\qs})}{2c_{1}}
\end{cases}$

$\triangleright$\textit{``w.p.$\propto$'' means ``with probability proportional
to''}}

\State{$\tilde{M}\leftarrow M_{0}+Z\cdot\mathrm{Geom}_{0}(1-e^{-2\qs})$}
\Comment{\textit{$\mathrm{Geom}_{0}(\cdot)$
is Geom RV over $\{0,1,\ldots\}$, generated using local randomness}}

\State{\Return$\mathrm{round}\big(\frac{\epsilon x}{\qs}+\tilde{M}+\mathrm{Unif}(-\frac{1}{2},\frac{1}{2})-U\big)$}

\Comment{\textit{$\mathrm{Unif}(\cdots)$
is generated using local randomness}}

\end{algorithmic}

\caption{\label{alg:enc}$\textsc{Encode}(x,\epsilon,\ell,\mathcal{P})$}
\end{algorithm}

\iflongver
\begin{algorithm}
\textbf{$\;\;\;\;$Input:} description $M$, parameters $\epsilon > 0$, $\ell > 1$, PRNG $\mathcal{P}$

\textbf{$\;\;\;\;$Output:} Laplace-noise-perturbed data $\hat{X}\in\mathbb{R}$

\smallskip{}

\begin{algorithmic}[1]

\State{Let $\delta_{0}>0$ be the solution to $e^{\delta_{0}}=\delta_{0}\ell+1$}

\State{Use PRNG $\mathcal{P}$ to sample $T\in\mathbb{Z}_{\ge0}$ with cdf given by \eqref{eq:F_T}}

\State{Use PRNG $\mathcal{P}$ to sample $U\sim\mathrm{Unif}(-1/2,1/2)$}

\State{\Return$2^{-T}\qs_0(M+U)/\epsilon $}

\end{algorithmic}

\caption{\label{alg:dec}$\textsc{Decode}(M,\epsilon,\ell,\mathcal{P})$}
\end{algorithm}
\fi

\section{Evaluation}
\begin{figure}
\begin{centering}
\ifshortver
\else
\includegraphics[scale=0.42]{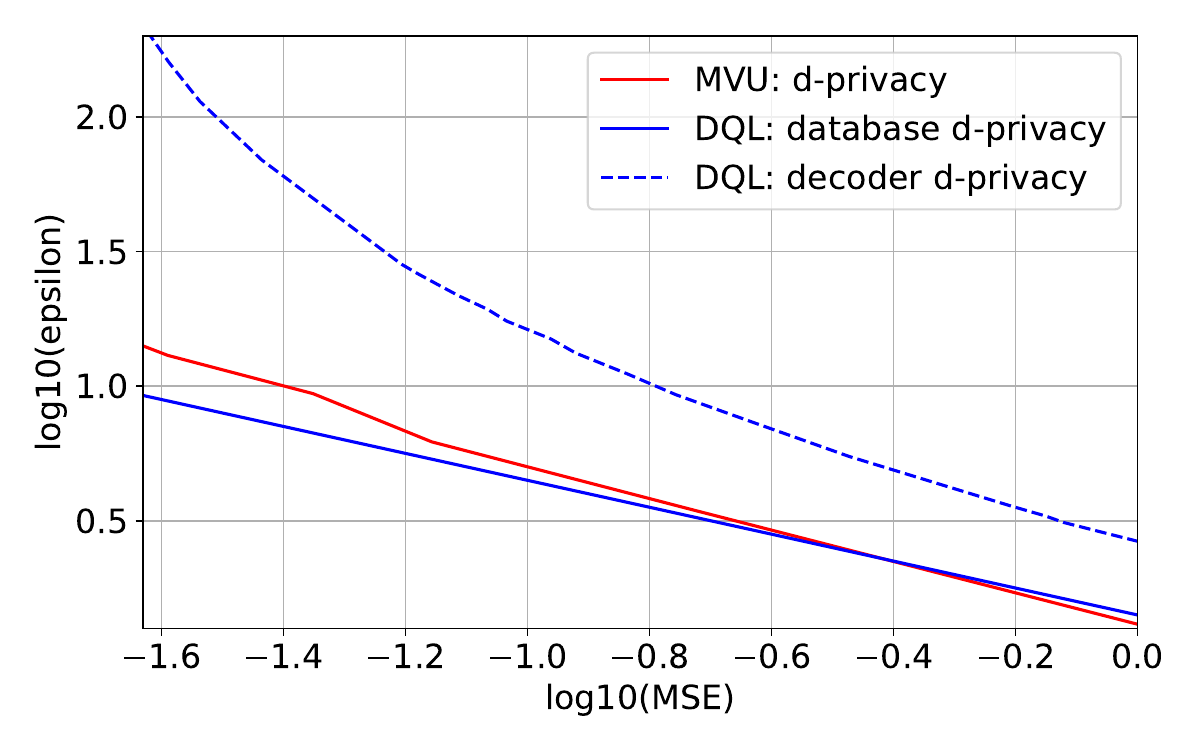}
\fi
\par\end{centering}
\caption{\label{fig:experiment}The logarithm of the $\epsilon$ for the (database/decoder) $\epsilon \cdot d$-privacy of MVU, and the database and decoder $\epsilon \cdot d$-privacy of DQL, versus the logarithm of the mean squared error.}
\end{figure}

Since the proposed DQL mechanism has the same behavior as the Laplace mechanism (except with a smaller communication cost), regarding the effect of DQL on learning and estimation tasks, readers are referred to the vast literature on the Laplace mechanism. Here we present an experiment where the data $x \sim \mathrm{Unif}(-1,1)$, and we plot the $\epsilon$ for the decoder and database $\epsilon \cdot d$-privacy of DQL against the mean squared error $\mathbb{E}[|\hat{X}-x|^2]$ in Figure~\ref{fig:experiment} in log scale.
\footnote{Source code: https://github.com/shmoradiali/DyadicQuantizedLaplace} 
We compare to the minimum variance unbiased (MVU) mechanism \cite{chaudhuri2022privacy} (where the decoder and database privacy parameters are the same), which computes the optimal unbiased random quantizer with a deterministic finite set of $\hat{X}$ (and hence is not Laplace) via nonconvex optimization. Both mechanisms are fixed to use $5$ bits per sample on average.\footnote{Elias gamma code~\cite{elias1975universal} is used for DQL. The expected length is set to $5$.} The database privacy of DQL is better for small MSE, though the decoder privacy is not as good, making DQL suitable when database privacy is more important. Compared to MVU which involves nonconvex optimization with size exponential in the number of bits (computationally expensive for $>5$ bits per sample), the running time of DQL is only linear in the number of bits, making DQL a lightweight substitute for the Laplace mechanism. Also, unlike MVU, DQL can work on unbounded input $x \in \mathbb{R}$.







\iflongver
\section{Conclusion and Discussions}
In this paper, we presented the dyadic quantized Laplace (DQL) mechanism, which is the first mechanism realizing the Laplace mechanism precisely (i.e., the total noise added to the private data owned by the users follows the Laplace distribution), but with a finite amount of communication bits (instead of sending real numbers) and a provable differential privacy guarantee (in the sense of metric privacy \cite{Chatzikokolakis13Dprivacy}) against the decoder and database.
Furthermore, we did several experiments to compare our proposed mechanism against the MVU mechanism \cite{chaudhuri2022privacy}, and the experiment results show that DQL has better database privacy guarantee  for small mean squared error, even though the decoder privacy is not as good. Hence, when database privacy is vital in the system design, DQL is a better option compared to existing candidates.

A future research direction is to extend this construction to other noise distributions, e.g., the Gaussian mechanism \cite{Dwork14Book}, and the Laplace mechanism over $\mathbb{R}^n$ for the $d$-privacy with the $\ell_2$ distance instead of $\ell_1$ \cite{Chatzikokolakis13Dprivacy}.

\fi

\bibliographystyle{IEEEbib}
\bibliography{ref}

\ifshortver
\else


\section{Appendix: Proof of Theorem~\ref{THM:PRIVACY}}\label{subsec:privacy_pf}

Assume $b=1$ without loss of generality. We have $1\cdot d$-privacy
against the database since the noise distribution is $\mathrm{Laplace}(0,1)$.
We now prove the bound for decoder privacy. Fix any $t\in\mathbb{Z}_{\ge0}$.
Let $\delta=\delta_{t}=2^{-t}\delta_{0}$. The function $F_{T}(t)g_{\delta}(x)$
is piecewise linear with vertices in the form
\[
\left(k\delta,\;\frac{F_{T}(t)}{c_{\delta}}e^{-|k\delta|}\right)
\]
for $k\in\mathbb{Z}$. The function $F_{T}(t-1)g_{2\delta}(x)$ is
piecewise linear with vertices in the form
\[
\left(k\delta,\;\frac{F_{T}(t-1)}{c_{2\delta}}e^{-|k\delta|}\right)
\]
for even $k$. Hence the function $\tilde{f}_{t}(x)=F_{T}(t)g_{\delta}(x)-F_{T}(t-1)g_{2\delta}(x)$
(which $f_{t}(x)$ in \eqref{eq:f_T} is proportional to) is piecewise linear
with vertices in the form
\[
\left(k\delta,\;\frac{F_{T}(t)}{c_{\delta}}e^{-|k\delta|}-\frac{F_{T}(t-1)}{c_{2\delta}}e^{-|k\delta|}\right)
\]
for even $k$, and
\[
\left(k\delta,\;\frac{F_{T}(t)}{c_{\delta}}e^{-|k\delta|}-\frac{F_{T}(t-1)}{c_{2\delta}}\cdot\frac{e^{-|(k-1)\delta|}+e^{-|(k+1)\delta|}}{2}\right)
\]
for odd $k$. Let
\begin{align*}
r & =F_{T}(t-1)/F_{T}(t)\\
 & =\frac{4-4(\delta\ell+1)e^{-\delta}}{(1+e^{-\delta})^{2}\big(2(1+e^{-2\delta})^{-1}-\delta\ell-1\big)}.
\end{align*}
By the arguments in Lemma \ref{lem:piecewise_lip}, over the interval $x\in[0,\delta]$, $\delta\log f_{t}(x)$
has a Lipschitz constant
\begin{align*}
 & \frac{\frac{F_{T}(t)}{c_{\delta}}-\frac{F_{T}(t-1)}{c_{2\delta}}}{\frac{F_{T}(t)}{c_{\delta}}e^{-\delta}-\frac{F_{T}(t-1)}{c_{2\delta}}\cdot\frac{1+e^{-2\delta}}{2}}-1\\
 & =\frac{\frac{1}{c_{\delta}}-\frac{r}{c_{2\delta}}}{\frac{1}{c_{\delta}}e^{-\delta}-\frac{r}{c_{2\delta}}\frac{1+e^{-2\delta}}{2}}-1.
\end{align*}
Over the interval $x\in[\delta,2\delta]$, $\delta\log f_{t}(x)$
has a Lipschitz constant being the maximum of
\begin{align*}
 & \frac{\frac{F_{T}(t)}{c_{\delta}}e^{-2\delta}-\frac{F_{T}(t-1)}{c_{2\delta}}e^{-2\delta}}{\frac{F_{T}(t)}{c_{\delta}}e^{-\delta}-\frac{F_{T}(t-1)}{c_{2\delta}}\cdot\frac{1+e^{-2\delta}}{2}}-1\\
 & \le\frac{\frac{1}{c_{\delta}}-\frac{r}{c_{2\delta}}}{\frac{1}{c_{\delta}}e^{-\delta}-\frac{r}{c_{2\delta}}\frac{1+e^{-2\delta}}{2}}-1,
\end{align*}
and
\begin{align*}
 & \frac{\frac{F_{T}(t)}{c_{\delta}}e^{-\delta}-\frac{F_{T}(t-1)}{c_{2\delta}}\cdot\frac{1+e^{-2\delta}}{2}}{\frac{F_{T}(t)}{c_{\delta}}e^{-2\delta}-\frac{F_{T}(t-1)}{c_{2\delta}}e^{-2\delta}}-1\\
 & =\frac{\frac{1}{c_{\delta}}-\frac{r}{c_{2\delta}}\cdot\frac{e^{\delta}+e^{-\delta}}{2}}{\frac{1}{c_{\delta}}e^{-\delta}-\frac{r}{c_{2\delta}}e^{-\delta}}-1\\
 & \le\frac{\frac{1}{c_{\delta}}-\frac{r}{c_{2\delta}}}{\frac{1}{c_{\delta}}e^{-\delta}-\frac{r}{c_{2\delta}}\frac{1+e^{-2\delta}}{2}}-1.
\end{align*}
The other intervals are similar. Hence, over $x\in\mathbb{R}$, $\delta\log f_{t}(x)$
has a Lipschitz constant
\begin{align*}
 & \frac{\frac{1}{c_{\delta}}-\frac{r}{c_{2\delta}}}{\frac{1}{c_{\delta}}e^{-\delta}-\frac{r}{c_{2\delta}}\frac{1+e^{-2\delta}}{2}}-1\\
 & =\frac{\frac{1}{\delta}\frac{1-e^{-\delta}}{1+e^{-\delta}}-\frac{r}{2\delta}\frac{1-e^{-2\delta}}{1+e^{-2\delta}}}{\frac{1}{\delta}\frac{1-e^{-\delta}}{1+e^{-\delta}}e^{-\delta}-\frac{r}{2\delta}\frac{1-e^{-2\delta}}{1+e^{-2\delta}}\frac{1+e^{-2\delta}}{2}}-1\\
 & =\frac{\frac{1-e^{-\delta}}{1+e^{-\delta}}-\frac{r}{2}\frac{(1+e^{-\delta})(1-e^{-\delta})}{1+e^{-2\delta}}}{\frac{1-e^{-\delta}}{1+e^{-\delta}}e^{-\delta}-\frac{r}{4}(1+e^{-\delta})(1-e^{-\delta})}-1\\
 & =\frac{1-\frac{r}{2}\frac{(1+e^{-\delta})^{2}}{1+e^{-2\delta}}}{e^{-\delta}-\frac{r}{4}(1+e^{-\delta})^{2}}-1\\
 & =\frac{1-\frac{2-2(\delta\ell+1)e^{-\delta}}{2-(\delta\ell+1)(1+e^{-2\delta})}}{e^{-\delta}-\frac{1-(\delta\ell+1)e^{-\delta}}{2(1+e^{-2\delta})^{-1}-\delta\ell-1}}-1\\
 & =\frac{\frac{(\delta\ell+1)(2e^{-\delta}-1-e^{-2\delta})}{2-(\delta\ell+1)(1+e^{-2\delta})}}{\frac{2e^{-\delta}(1+e^{-2\delta})^{-1}-1}{2(1+e^{-2\delta})^{-1}-\delta\ell-1}}-1\\
 & =\frac{(\delta\ell+1)(2e^{-\delta}-1-e^{-2\delta})}{2e^{-\delta}-(1+e^{-2\delta})}-1\\
 & =\delta\ell.
\end{align*}
Therefore $\log f_{t}(x)$ has a Lipschitz constant $\ell$. In the
dyadic quantized Laplace mechanism, we choose $T$ at random, and
run the quantized piecewise linear mechanism on $f_{T}(x)$. Since
each of the quantized piecewise linear mechanism is $\ell\cdot d$-private,
the overall mechanism is $\ell\cdot d$-private as well.

\fi

\end{document}